\documentclass[11pt]{article}
\usepackage[utf8]{inputenc}

\usepackage{amssymb, amsmath, amsthm, mathtools}

\RequirePackage{thm-restate}
\usepackage{fullpage}
\usepackage{graphicx}
\usepackage[dvipsnames]{xcolor}
\usepackage{tikz}
\usetikzlibrary {arrows.meta,bending,positioning}
\usepackage{xspace}
\RequirePackage[colorlinks=true]{hyperref}
\RequirePackage[capitalize,nameinlink,noabbrev]{cleveref}
\RequirePackage{xcolor}
\hypersetup{
	linkcolor={red!75!black},
	citecolor={blue!75!black},
	urlcolor={red!75!black},
}
\RequirePackage[hyperpageref]{backref}

\def\Z{\mathbb{Z}}
\def\F{\mathbb{F}}

\def\A{\mathcal{A}}
\def\B{\mathcal{B}}
\def\P{\mathcal{P}}
\def\AUX{\mathrm{AUX}}
\def\MAP{\mathrm{MAP}}
\def\TR{\mathrm{TR}}
\DeclareMathOperator{\poly}{poly}

\newcommand{\problem}[1]{\textsf{#1}}
\newcommand{\ThreeSUM}{\problem{3SUM}\xspace}
\newcommand{\ThreeSUMInd}{\problem{3SUM-Indexing}\xspace}
\newcommand{\kSUM}{\problem{kSUM}\xspace}
\newcommand{\kSUMInd}{\problem{kSUM-Indexing}\xspace}

\newcommand{\kXORInd}{\problem{kXOR-Indexing}\xspace}

\newcommand{\GSI}{\problem{Gapped String Indexing}\xspace}

\newcommand{\JI}{\problem{Jumbled Indexing}\xspace}

\let\Im\relax
\DeclareMathOperator{\Im}{Im}

\newcommand{\eps}{\varepsilon}
\renewcommand{\epsilon}{\varepsilon}

\theoremstyle{plain}

\newtheorem{theorem}{Theorem}[section]
\newtheorem{lemma}[theorem]{Lemma}

\newtheorem{conjecture}[theorem]{Conjecture}

\newtheorem{remark}[theorem]{Remark}

\theoremstyle{definition}
\newtheorem{definition}[theorem]{Definition}

\title{Improved Time-Space Tradeoffs for \ThreeSUMInd}
\author{
    Itai Dinur
    \thanks{Ben-Gurion University and Georgetown University. Email: \texttt{dinuri@bgu.ac.il}.}
    \and
	Alexander Golovnev%
	\thanks{Georgetown University. Email: \texttt{alexgolovnev@gmail.com}. Supported by the National Science Foundation CAREER award (grant CCF-2338730).}
}
\date{}

\begin{document}
\maketitle
\thispagestyle{empty}
\begin{abstract}
\ThreeSUMInd is a preprocessing variant of the \ThreeSUM problem that has recently received a lot of attention.
The best known time-space tradeoff for the problem is $T S^3 = n^{6}$ (up to logarithmic factors), where $n$ is the number of input integers, $S$ is the length of the preprocessed data structure, and $T$ is the running time of the query algorithm. This tradeoff was achieved 
in~\cite{KP19,GGHPV20} using the Fiat-Naor generic algorithm for Function Inversion. Consequently,~\cite{GGHPV20}
asked whether this algorithm can be improved by leveraging the structure of \ThreeSUMInd. 

In this paper, we exploit the structure of \ThreeSUMInd to give a time-space tradeoff of $T S = n^{2.5}$, which is better than the best known one in the range $n^{3/2} \ll S \ll n^{7/4}$. We~further extend this improvement to the \kSUMInd problem---a generalization of \ThreeSUMInd---and to the related \kXORInd problem, where addition is replaced with XOR.
Additionally, we improve the best known time-space tradeoffs for the \JI problem, which is a well-known data structure problem related to \ThreeSUMInd.

Our improvement comes from an alternative way to apply the Fiat-Naor algorithm to \ThreeSUMInd. Specifically, we exploit the structure of the function to be inverted by decomposing it into ``sub-functions'' with certain properties. This allows us to apply an improvement to the Fiat-Naor algorithm (which is not directly applicable to \ThreeSUMInd), obtained 
in~\cite{GGPS23}
in a much larger range of parameters. We believe that our techniques may be useful in additional application-dependent optimizations of the Fiat-Naor algorithm.  
\end{abstract}
\newpage
\setcounter{page}{1}.

\section{Introduction}
In the \ThreeSUM problem, the input is a set of integers $A = \{a_0,\ldots,a_{n-1} \}$, and the goal is to find $(i,j,k) \in [n]^3$  
(where $[n] = \{0,\ldots,n-1\}$) such that $a_i + a_j + a_k = 0$.
The modern version of the famous \ThreeSUM conjecture~\cite{GO95} asserts that the problem cannot be solved in time $n^{2 - \delta}$ for any constant $\delta > 0$. 
This conjecture was shown to imply many conditional lower bounds for geometric, combinatorial, and string search problems (see, e.g.,~\cite{V18} for an excellent survey on this topic).

\subsection{\ThreeSUMInd}
In this paper, we analyze a preprocessing variant of \ThreeSUM,
known as \ThreeSUMInd. This variant was first defined by Demaine and Vadhan~\cite{DV01} in an unpublished note, and then reconsidered by Goldstein,
Kopelowitz, Lewenstein and Porat~\cite{GKLP17}.

In the \ThreeSUMInd problem, the input is an array of positive integers $A = \{a_0,\ldots,a_{n-1}\}$. 
An algorithm $\A$ for the problem is a pair $\A = (\A_0,\A_1)$, where $\A_0$ is the preprocessing algorithm and~$\A_1$ is the online algorithm. 
The preprocessing algorithm $\A_0$ receives $A$ as input and outputs a data structure of length $S$ bits.\footnote{One can also define the problem by allowing $\A_0$ to output $S$ words of size poly-logarithmic in $n$. Since we ignore such poly-logarithmic factors in this paper, we consider these definitions equivalent.} 
The online algorithm $\A_1$ receives as input an integer challenge $y$ 
and has access to the bits of the data structure output by $\A_0$. Its goal is to output a pair $(i,j) \in [n]^2$ such that $a_i + a_j = y$, if such a pair exists, and $\bot$ otherwise. 
In the cell-probe model, 
the time complexity of $\A_1$ (denoted by $T$) is only measured by the number of data structure bits that it queries. One can also consider other computational models which account for the actual runtime of $\A_1$ (e.g., in the RAM model) or the preprocessing time. 

There are two trivial algorithms for \ThreeSUMInd: the first one stores $A$ sorted, while on input $y$, $\A_1$ searches this array for $y - a_i$ for each $i \in [n]$. This algorithm has 
$S = \widetilde{O}(n)$ and 
$T = \widetilde{O}(n)$
(where the $\widetilde{O}$ notation
suppresses factors poly-logarithmic in the input length). 
The second algorithm stores the sorted sumset $A + A = \{ a_i + a_j \mid (i,j) \in [n]^2 \}$, while on input $y$, $\A_1$ searches this sumset. The second algorithm has 
$S = \widetilde{O}(n^2)$ and 
$T = \widetilde{O}(1)$.

Consequently,~\cite{DV01,GKLP17} formulated three conjectures:
\begin{conjecture}[\cite{GKLP17}]
\label{conj:1}
If there exists an algorithm which solves \ThreeSUMInd with prepro-
cessing space $S$ and $T = \widetilde{O}(1)$ probes then $S = \widetilde{\Omega}(n^2)$.
\end{conjecture}

\begin{conjecture}[\cite{DV01}]
\label{conj:2}
If there exists an algorithm which solves \ThreeSUMInd with preprocessing
space $S$ and $T$ probes, then $ST = \widetilde{\Omega}(n^2)$.
\end{conjecture}

\begin{conjecture}[\cite{GKLP17}]
\label{conj:3}
If there exists an algorithm which solves \ThreeSUMInd with $T =
\widetilde{O}(n^{1 - \delta})$ probes for some $\delta > 0$ then 
$S = \widetilde{\Omega}(n^2)$.
\end{conjecture}
These conjectures are in ascending order of strength:
\begin{align*}
\text{Conjecture~\ref{conj:3}} \Rightarrow  
\text{Conjecture~\ref{conj:2}}
\Rightarrow  
\text{Conjecture~\ref{conj:1}}.
\end{align*}

In terms of lower bounds,
Demaine and Vadhan proved that any 1-probe data structure
for \ThreeSUMInd requires space $\widetilde{\Omega}(n^2)$, and left the case of $T > 1$ open.
Then~\cite{GGHPV20} 
proved that for every non-adaptive algorithm that uses space $S$ and query time $T$ and solves \ThreeSUMInd, it holds that $S = \widetilde{\Omega}(n^{1+1/T})$.
More recently, Chung and Larsen~\cite{CL23} proved similar bounds for adaptive algorithms.
All these lower bounds assume the data structure consists of $S$ words of size $\widetilde{O}(1)$.
Since proving super-logarithmic query-time bounds for static data structures even with space
$S = O(n)$ is a major
open problem, the lower bounds of~\cite{CL23} are essentially the best possible (barring a significant breakthrough). 

From the algorithmic side, \cref{conj:3} was refuted by Kopelowitz and Porat~\cite{KP19} and by~\cite{GGHPV20}
using the same techniques. Specifically,~\cite{KP19,GGHPV20} describe an algorithm for \ThreeSUMInd with $T =
\widetilde{O}(n^{\delta})$ and
$S = \widetilde{\Omega}(n^{2 - \delta/3})$. 
The algorithm is based on the classical Fiat-Naor algorithm~\cite{FN91} for the Function Inversion problem,
where the goal is to invert an efficiently computable function $f: [N] \to [N]$ by a two-phase algorithm: a preprocessing algorithm for $f$ that outputs an advice string of length $S$ bits, and an online algorithm that receives a challenge $y \in [N]$ and finds $x \in [N]$ such that $f(x) = y$ after making $T = \tilde{O}(N^3/S^3)$ queries to~$f$. In the case of \ThreeSUMInd, by hashing, we may essentially assume that the input $a_0,\ldots,a_{n-1}$ satisfies $a_i \in [n^2]$ for every $i \in [n]$. We then define the function $f:[n^2] \to [n^2]$ by $f(i,j) = a_i + a_j$ and apply the Fiat-Naor algorithm to $f$. 

The Fiat-Naor based algorithm of~\cite{KP19,GGHPV20} is currently the best known algorithm for \ThreeSUMInd. 
Recently, Bille et al.~\cite{BGLPRS24} proved that  well-studied problems in data structures (specifically, \GSI and \JI) reduce to \ThreeSUMInd. 
Thus, the algorithm of~\cite{KP19,GGHPV20} is used to obtain the best known algorithms for these problems. The \ThreeSUMInd problem has also inspired more applications of the Fiat-Naor algorithm to related data structure problems~\cite{AESZ23,AHY25}. 
Given the extensive research involving \ThreeSUMInd and its variants, it is natural to ask whether it is possible to 
exploit the structure of the problem to improve the Fiat-Naor based algorithm of~\cite{KP19,GGHPV20}.
This question was posed explicitly in~\cite[Open Question 2]{GGHPV20} and \cite[Lecture 14]{VW24}.

We also mention the recent improvement of the Fiat-Naor algorithm for Function Inversion by~\cite{GGPS23}. 
Unfortunately, this improvement is in a restricted parameter range, and when applied to \ThreeSUMInd, the algorithm of~\cite{GGPS23} is outperformed by the trivial algorithm with $S = \widetilde{O}(n), \,T = \widetilde{O}(n)$.

\subsection{Application-Dependent Improvements of the Fiat-Naor Algorithm}
Looking at the open question of~\cite{GGHPV20} from a broader perspective, we recall that
the Fiat-Naor algorithm was originally designed to invert ``unstructured'' cryptographic one-way functions. 
The fact that it gives the best known time-space tradeoff for a broad spectrum of \emph{structured} problems ~\cite{CK19,KP19,GGHPV20,AESZ23,BGLPRS24,ACDI24,M24} is surprising. Quoting \cite{ACDI24}: 
\begin{quote}
It is a striking fact that, for each of these problems (\ThreeSUMInd, \problem{Collinearity Indexing}, \GSI), the best known space-time trade-offs
are achieved using such a general tool as the Fiat-Naor inversion scheme.    
\end{quote}

More generally, the Function Inversion problem has found major applications across multiple areas such as cryptanalysis~\cite{H80,BS00,BSW01,O03,NS05}, circuit and data structure lower bounds~\cite{Y90,CK19,DKKS21}, algorithms~\cite{KP19,GGHPV20,AESZ23,BGLPRS24,ACDI24,M24}, information theory~\cite{DKKS21}, and most recently even meta-complexity~\cite{MP24,HIW24}.
On the other hand, we are not aware of any previous application-dependent improvement of the Fiat-Naor scheme. Thus,
finding such improvements is an interesting research topic. Informally, the challenge is due to the unique data structure used by the Fiat-Naor algorithm, which is built by iterating the function we wish to invert. Such iterations seem to destroy most properties that can be exploited to optimize the algorithm (e.g., iterating a function represented by a low-degree polynomial quickly increases its degree).

\subsection{Our Results}
In this paper, we improve the best known time-space tradeoff for \ThreeSUMInd~\cite{KP19,GGHPV20} for the parameter range $n^{3/2} \ll S \ll n^{7/4}$.
In particular, we obtain the first application-dependent improvement of the Fiat-Naor scheme.

\begin{figure}[!ht]
  \centering
  \begin{tikzpicture}[samples=200,xscale=6.4,yscale=3.2]
  \draw[black,dotted,step=0.5] (2.3,0) grid (0.99,1.3);
  \foreach \x/\t in {1/$1$,{3/2}/{$\frac{3}{2}$},{5/3}/{$\frac{5}{3}$},{7/4}/{$\frac{7}{4}$},2/$2$} \node at (\x,-0.1) {\footnotesize \t};
  \foreach \y in {0,1 } \node at (0.95,\y) {\footnotesize $\y$};
  \node at (2.3,-0.2) {$\frac{\log{S}}{\log{n}}$};
  \node at (0.9,1.3) {$\frac{\log{T}}{\log{n}}$};
  \draw[->, >={Stealth[length=2mm]}, draw opacity=0] (1,0) -- (1,1.3);
  \draw[->, >={Stealth[length=2mm]}, draw opacity=0] (1,0) -- (2.3,0);
  \draw[green,line width=2,dotted] (1,1)--(2,1);
  \draw[red,line width=2] (1.5,1)--(1.75,0.75);
  \draw[blue,line width=2,dashed] (5/3,1)--(2,0);
  \draw[black,dotted](5/3,0)--(5/3,1);
  \draw[black,dotted](7/4,0)--(7/4,3/4);
  \end{tikzpicture}
  \caption{The parameters of the data structures for \ThreeSUMInd are as follows: the trivial algorithm is represented by the \textcolor{green}{dotted green curve}, the Fiat-Naor-based algorithm~\cite{KP19,GGHPV20} is represented by the \textcolor{blue}{dashed blue curve}, and our algorithm is represented by the \textcolor{red}{solid red curve}.}\label{figure:1}
\end{figure}
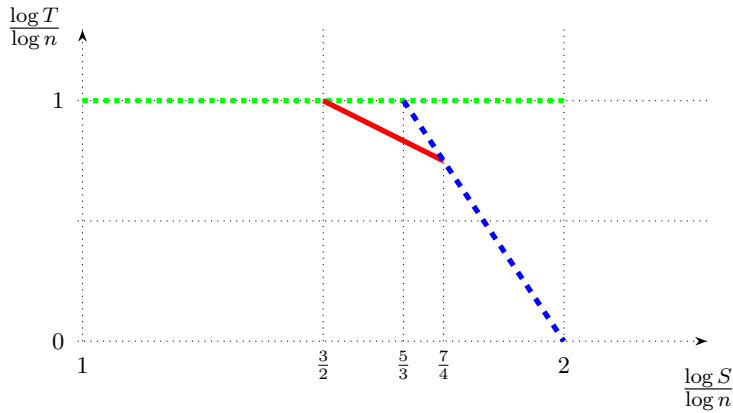

Specifically, 
we prove the following theorem.
\begin{restatable}{maintheorem}{ThreeSumThm}
\label{thm:threeSum}
For every $0\leq \delta\leq 1$, there is an $(S, T)$-algorithm for \ThreeSUMInd with space $S=\widetilde{O}(n^{2.5-\delta})$ and query time $T=\widetilde{O}(n^{\delta})$.
\end{restatable}
Our time-space tradeoff is compared to the previously best known one in \cref{figure:1}.
In particular, for 
$S = n^{5/3}$, we obtain 
$T = \widetilde{O}(n^{5/6})$,  while the previous best algorithm was the trivial one, which obtains $T =  \widetilde{O}(n)$.

We remark that the runtime of the online algorithm $\A_1$ in our algorithm in the standard RAM model is linear in its number of queries to the data structure, up to poly-logarithmic factors. Moreover, the runtime of the preprocessing algorithm $\A_0$ is $\widetilde{O}(n^2)$. (An algorithm with truly subquadratic preprocessing time $\widetilde{O}(n^{2-\eps})$ would refute the \ThreeSUM conjecture.)

We further apply a variant of our algorithm to the \kSUMInd problem, introduced in~\cite{GGHPV20}. In this problem, the input of the preprocessing algorithm is an array of positive integers $A = \{a_0,\ldots,a_{n-1} \}$, and the online algorithm for an integer challenge $y$ finds $(i_1,\ldots,i_{k-1}) \in [n]^{k-1}$ such that $a_{i_1} + \ldots + a_{i_{k-1}} = y$. \kSUMInd is a preprocessing version of \kSUM and a natural generalization of \ThreeSUMInd.

Once again, we focus on the parameter regime with sub-linear time. 
\begin{restatable}{maintheorem}{KSumThm}
\label{thm:kSum}
For every $k\geq3$ and every $0\leq \delta\leq1$, there is an $(S, T)$-algorithm for \kSUMInd with space $S=\widetilde{O}(n^{k-0.5-\delta})$ and query time $T=\widetilde{O}(n^{\delta})$.
\end{restatable}
We remark that there is a trivial algorithm for \kSUMInd with $S = \widetilde{O}(n^{k-2})$ and $T = \widetilde{O}(n)$, while an application of the Fiat-Naor algorithm~\cite{GGHPV20} gives $S = \widetilde{O}(n^{k - 1 -\delta/3})$ and $T = \widetilde{O}(n^{\delta})$. Our algorithm is better than the previous algorithms in the range $n^{k - 1.5} \ll S \ll n^{k - 1.25}$.
We further note that there are other trivial algorithms for \kSUMInd in the regime of super-linear~$T$, and that our algorithm can be easily extended to this regime too (as done in~\cite{GGHPV20} by taking larger values of $\delta$).

Next, we adapt our algorithm to the \kXORInd problem. In \kXORInd, the input to be preprocessed is an array of $n$ vectors in $\{0,1\}^\ell$. Then, given a query vector $y\in\{0,1\}^\ell$, the goal is to check if there are $(k-1)$ input vectors whose XOR is~$y$.

\begin{restatable}{maintheorem}{kXORThm}
\label{thm:kXor}
For every $k\geq3$ and every $0\leq \delta\leq1$, there is an $(S, T)$-algorithm for \kXORInd with space $S=\widetilde{O}(n^{k-0.5-\delta})$ and query time $T=\widetilde{O}(n^{\delta})$.
\end{restatable}
As an additional application, we use the reduction of Bille et al.~\cite{BGLPRS24} from the \JI problem to \ThreeSUMInd, and conclude an immediate improvement for it as well.

In the \JI problem (also known as \problem{Histogram Indexing}), the input to be preprocessed is a string~$S$ of length~$n$ over an alphabet~$\Sigma$. Given a query histogram $h\in\Z_{\geq0}^{|\Sigma|}$ (i.e., the number of occurrences of each letter from $\Sigma$), the task is to check if $S$ contains a substring whose histogram matches~$h$.

In the case of the binary alphabet $|\Sigma|=2$, Cicalese et al.~\cite{CFL09} gave an efficient algorithm with $S=O(N)$ and $T=O(1)$ solving \JI. For larger alphabets, Kociumaka, Radoszewski, and Rytter~\cite{KRR13} designed an algorithm that solves \JI in space $S=\widetilde{O}(n^{2-\delta'})$ and time $T=\widetilde{O}(n^{(2|\Sigma|-1)\delta'})$. Then, Chan and Lewenstein~\cite{CL15} improved the bound to $S=\widetilde{O}(n^{2-\delta'})$ and $T=\widetilde{O}(n^{(|\Sigma|+1)\delta'/2})$ by utilizing algorithms for a variant of \ThreeSUM. Finally, \cite{BGLPRS24} used a reduction from \JI to \ThreeSUMInd, together with the \ThreeSUMInd algorithm of \cite{KP19,GGHPV20}, resulting in an upper bound of $S=\widetilde{O}(n^{2-\delta'/3})$ and $T=\widetilde{O}(n^{\delta'})$. This bound improves on the previous bounds for all alphabets of size $|\Sigma|>5$. We further improve algorithms for \JI, achieving $S=\widetilde{O}(n^{2.5-\delta})$ and $T=\widetilde{O}(n^{\delta})$. This algorithm outperforms~\cite{BGLPRS24} for all $n^{3/2} \ll S \ll n^{7/4}$, and it improves on~\cite{CL15} for all $|\Sigma|>1 + 1/\delta'$.

\begin{restatable}{maincorollary}{JICor}
\label{cor:ji}
For every $0\leq \delta\leq 1$ and an alphabet $\Sigma$ of constant size $|\Sigma|=O(1)$, there is an $(S, T)$-algorithm for \JI with space $S=\widetilde{O}(n^{2.5-\delta})$ and query time $T=\widetilde{O}(n^{\delta})$.
\end{restatable}

\section{Technical Overview}
In this section we give a technical overview of our work. 
Since it is based on 
previous algorithms for the Function Inversion problem, we start by recalling them.

\subsection{The Hellman and Fiat-Naor Algorithms for Function Inversion}

\paragraph{Hellman's algorithm.} The research of non-uniform cryptanalytic time-space tradeoffs was initiated by Hellman~\cite{H80} in 1980. Hellman considered the problem of inverting a  function $f: [N] \to [N]$ using a two-phase algorithm. The preprocessing algorithm has full access to $f$ and computes an advice string of $S$ bits. The online algorithm receives as input the advice string and a challenge $y \in [N]$. In addition, the online algorithm is allowed to make $T$ oracle queries to~$f$. Its goal is to find a preimage $x \in f^{-1}(y)$ (if such a preimage exists).

Hellman gave a heuristic algorithm for the Function Inversion problem with a time-space tradeoff of $T S^2 = \widetilde{O}(N^2)$, assuming that $f$ is chosen uniformly at random. For parameters $s$ and $t$, the main data structure computed by the preprocessing algorithm is a table constructed via $s$ chains. Each chain starts from a uniformly chosen point $x \in [N]$ and is computed via $t$ iterative calls to $f$, where only $x$ and the endpoint $f^{(t)}(x)$ are stored in the table.\footnote{We use the notation $f^{(t)}$ to represent $f$ composed with itself
$t$ times.} The pairs of start and endpoints are sorted according to the endpoints. 

Given this table and a point $y \in [N]$ to invert, the online algorithm computes up to $t$ iterations of $f$ starting from $y$, and checks if any of them is an endpoint of a chain stored in the table. If~it reaches such an endpoint, it restarts the computation of the chain from the corresponding start point, aiming to reach $y$, thus successfully inverting it in time $\widetilde{O}(t)$. 

The online algorithm succeeds in inverting $y$ if it is covered by one of the chains in the table. Unfortunately, for a uniformly chosen function, one can cover only a small fraction of the image points of $f$ with a single table (assuming $s \ll N$). More specifically, if $s t^2 \approx N$, and we have already computed $s$ chains of length $t$, by the birthday paradox, an additional chain of length $t$ collides with a previous one with high probability. Thus, additional chains do not add much to the coverage of the table.
Consequently, a single table only covers about
$s \cdot t \approx N/t$ image points using $\widetilde{O}(s)$ space.

Hellman's (heuristic) solution was to compute about $t$ such tables, each computed with a different variant of $f$, defined by composing it with a simple permutation, such that inverting a variant of $f$ is equivalent to inverting $f$. 
Heuristically, these $t$ tables cover most of the image points of~$f$. Thus, the preprocessing advice consists of $t$ tables of space about $s$ bits and the online algorithm searches (essentially) each of them in time $t$. Overall, we have 
$S \approx st$ and $T \approx t^2 $, giving 
$T S^2 \approx t^2 \cdot (s t)^2 \approx N^2$.

\paragraph{The Fiat-Naor algorithm.} Fiat and Naor~\cite{FN91} made Hellman's algorithm rigorous by composing $f$ with a $k$-wise independent function for an appropriate choice of $k$ (the details are not important for this paper). More generally, Fiat and Naor considered the problem of inverting an arbitrary function $f: [N] \to [N]$ with collision probability $C(f) := \Pr_{x,x' \sim [N]} [f(x) = f(x')]$, and extended the time-space tradeoff for such a function to $T S^2 = \widetilde{O}(N^3 \cdot C(f))$. 

For functions with a large collision probability, one can do better by artificially decreasing the effective collision probability. Specifically, the advice string additionally consists of (roughly) $S$ points sampled uniformly at random, as well as their images. 
Using an appropriately chosen pseudorandom function (that is shared between the preprocessing and online algorithms),
these images are then bypassed when iterating $f$ for computing the chains in the tables. The online algorithm first checks if $y$ is contained in the image set in the advice string, and if so, it outputs the corresponding preimage. Otherwise, it tries to invert $y$ using the tables. 

Observe that with high probability, all images whose preimage size is at least $\widetilde{\Omega}(N/S)$ are included in the image set. Thus, with high probability, the collision probability of $f$ is effectively reduced to at most $\widetilde{O}(1/S)$.
This leads to a worst-case time-space tradeoff of
$T S^2 = \widetilde{O}(N^3/S)$ (i.e.,
$T S^3 = \widetilde{O}(N^3)$) for inverting any function.

\paragraph{The improvement by~\cite{GGPS23}.}
It was observed in~\cite{GGPS23} that the Fiat-Naor worst-case time-space tradeoff can be improved. 
Specifically,~\cite{GGPS23} proposed an alternative to the preprocessing algorithm that originally included the set of uniformly chosen points and their images in the advice string. In this alternative, using shared randomness, the preprocessing and online algorithms compute such a set of size (roughly) $T$ from a shared random seed. In the worst-case, this effectively reduces the collision probability of $f$ to about $1/T$ and results in a time-space tradeoff of 
\begin{align}
\label{eq:improve}
 T S^2 = \widetilde{O}(N^3/T)   
\end{align}
(i.e., $T S = \widetilde{O}(N^{3/2})$). This improves upon the Fiat-Naor tradeoff in the case when $S \ll T$ (or $S \ll N^{3/4}$). On the other hand, we note that the computation of the set of points and their images from the shared random seed is non-uniform (i.e., it is not efficient in the standard RAM model).

We further remark that when $T \gg S$, the online algorithm of~\cite{GGPS23} has space complexity of $\widetilde{\Omega}(T) \gg S$. While this online space is ignored in the preprocessing model of computation, this can be viewed as a disadvantage of the algorithm. 

\subsection{Our Techniques}
Recall that the best known worst-case time-space tradeoff for \ThreeSUMInd is 
$T S^3 = \widetilde{O}(n^6)$. It was derived in~\cite{KP19,GGHPV20} by applying the Fiat-Naor tradeoff to the function $f(i,j) = a_i+a_j$ (one may assume that $f:[n^2] \to [n^2]$ by standard hashing techniques). 
Moreover, there is a trivial time-space tradeoff of $T = S = \widetilde{O}(n)$, while clearly $S = \widetilde{\Omega}(n)$ must hold for any algorithm that succeeds answering all the queries.
Thus, the only relevant parameter range is $T \ll n$ (and $S = \widetilde{\Omega}(n)$). 
Consequently, the improvement of~\cite{GGPS23} to the Fiat-Naor tradeoff (effective only when $S \ll T$) is not directly applicable to \ThreeSUMInd. 

\paragraph{An initial improvement in the cell-probe model.}
A closer look reveals that a variant of the improvement of~\cite{GGPS23} is applicable in the cell-probe model, where we allow the online algorithm to be inefficient as long as it \emph{reads} only a few positions of the advice.
Namely, we measure its complexity only by the number of queries it makes to the advice string (each query reads a word of a poly-logarithmic number of bits). 

Specifically, in order to reduce the effective collision probability of $f$, both the preprocessing and online algorithms sample the same uniform subset of $A$, denoted $A'$, of size (about) $T$. Now the online algorithm will correctly answer all queries from the set $A' + A'$ at no cost (as this does not require reading any part of the advice). Next, we again define $f(i,j)=a_i+a_j$ with one modification: whenever $a_i+a_j$ falls into $A'+A'$, $f$ just outputs a pseudorandom value. With high probability, this reduces  the (worst-case) collision probability of~$f$ to about $1/T^2$, resulting in a time-space tradeoff of about $T S^2 = \widetilde{O}(n^6/T^2)$, or $T^3 S^2 = \widetilde{O}(n^6)$.
This tradeoff is better than the known one 
$T S^3 = \widetilde{O}(n^6)$ when
$S \ll T^2$ (or $S \ll n^{12/7}$).

\paragraph{Our algorithm.}
We devise an improved algorithm with the tradeoff of $T S = \widetilde{O}(n^{2.5})$. 
This tradeoff is equivalent to $T^2 S^2 = \widetilde{O}(n^{5})$, and since $T \ll n$, it is always better than the previous one of $T^3 S^2 = \tilde{O}(n^6)$ for all relevant parameter settings. Moreover, it is obtained by an \emph{efficient} uniform online algorithm, rather than  in the cell-probe model. Since this algorithm strictly improves upon the cell-probe algorithm, we will not consider the cell-probe algorithm in the remainder of this paper.

We now sketch the details of our improved algorithm.
The improvement still involves applying the technique of~\cite{GGPS23}. However, we first derive an alternative way to apply the Fiat-Naor algorithm to \ThreeSUMInd which ``breaks down'' the function $f(i,j) = a_i + a_j$ into about $n$ ``sub-functions'' with domain and range of size about $n$. 
The preprocessing algorithm will apply the Fiat-Naor preprocessing algorithm to each such function independently (with space reduced by a factor of $n$).

Given a query $y \in A + A$, the task of the online algorithm will be reduced to inverting only a single such sub-function. Moreover, given the advice string, each sub-function will be efficiently computable. 

Since both the space and the range size of each sub-function is reduced by a factor of about~$n$, 
while only one function is inverted online (and it is efficiently computable), 
we can obtain the time-space tradeoff $T (S/n)^3 = \widetilde{O}(n^3)$, or $T S^3 = \tilde{O}(n^6)$. This only recovers the known Fiat-Naor tradeoff for \ThreeSUMInd. Yet, now the improvement of~\cite{GGPS23} will be much more noticeable: the task of the online algorithm is reduced to inverting a single function with range size of about $n$, giving the tradeoff $T (S/n)^2 = \tilde{O}(n^3/T)$, or $T S = \tilde{O}(n^{2.5})$. This tradeoff is similar to~\cref{eq:improve}, with the space divided by $n$ (as each sub-function is preprocessed separately). 

Interestingly, we also observe that unlike the original application of~\cite{GGPS23} to the Function Inversion problem (where the online space complexity exceeded the advice string length),
the space complexity of our online algorithm is still $\widetilde{O}(S + T) = \widetilde{O}(S)$ (as $S \gg T$).
For a similar reason, unlike~\cite{GGPS23}, our algorithm does not require any shared randomness between the preprocessing and online algorithms, while running efficiently in the RAM model. 
Specifically, rather than a seed that needs to be expanded inefficiently, our advice string includes the actual set of uniformly chosen points for each sub-function. Moreover, as all sub-functions can share the same set (up to a simple translation), we include only a single set. In our setting, this has negligible overhead in terms of the length of the advice string.

It remains to show how the sub-functions are defined. In order to fulfill the above constraints, the definition crucially relies on the additive structure of the problem. 
In particular, the preprocessing algorithm draws two random primes $p,q = \widetilde{\Theta}(n)$, so that it is sufficient to solve \ThreeSUMInd $\bmod \,pq$. 
For $d \in [q]$, the function $f_d: [n] \to [p]$ only iterates over pairs $(i,j) \in [n]^2$ such that $a_i + a_j \equiv d \bmod q$. On input $i \in [n]$, it is defined by computing (the first) $j \in [n]$ such that $a_{i} + a_{j} \equiv d \bmod q$, and returning $(a_{i} + a_{j}) \bmod p$ (if such $j$ does not exist, it returns a predefined  value chosen at random). Thus, inverting $y \in A + A$ is reduced to computing $f^{-1}_{y \bmod q}(y \bmod p)$. In order for each $f_d$ to be efficiently computable, the advice string also includes a sorted array consisting of the values $a_1 \bmod q, \ldots, a_n \bmod q$. 
Thus, on input $i \in [n]$, (the first) $j \in [n]$ such that $a_{i} + a_{j} \equiv d \bmod q$ is computed by binary search for $d - a_{i} \bmod q$.

The idea of improving the Fiat-Naor algorithm by defining appropriate efficiently-computable sub-functions is rather generic, and may find additional applications besides 
\ThreeSUMInd and related problems with additive structure. Thus, we first devise a general improved algorithm for Function Inversion, assuming that the inverted function can be broken down into sub-functions with certain constraints. We then apply this algorithm to \ThreeSUMInd.

\section{Preliminaries}
We denote the sets of integers and non-negative integers by $\Z$ and $\Z_{\geq0}$, respectively.
For a positive integer $m$, $[m]$ denotes the set of integers $\{0,\ldots,m-1\}$. 
For two vectors $x,y\in\{0,1\}^\ell$, by $x\oplus y\in\{0,1\}^\ell$ we denote their bitwise XOR.
For two sets of integers $A,B\subset\Z$, we use $A+B$ to denote the set of pairwise sums $A+B=\{a+b\colon\, a\in A, b\in B\}$.  
For a positive integer $k$ and a set $A\subset\Z$ of integers, by $kA$ we denote the set of all $k$-wise sums of elements from~$A$: 
\[
kA = \{a_1+\ldots+a_k\colon\, a_1,\ldots,a_k\in A \}\;.
\]
For a function $f\colon[N]\to[N]$, $\Im(f)$ denotes the image of~$f$, and $f^{-1}(y)=\{x\colon\,f(x)=y\}$. 

We use $\log(\cdot)$ to denote the logarithm base~2, i.e., $\log(2^n) = n$, and we use $\ln(\cdot)$ to denote the natural logarithm. 

The $\widetilde{O}(\cdot)$ and $\widetilde{\Omega}(\cdot)$ notations suppress factors poly-logarithmic in the input length. For example, for an input $A=(a_1,\ldots, a_n)\in\Z_{\geq0}^n$, where $M=\max_i{a_i}$, the running time $\widetilde{O}(n^\delta)$ stands for $n^\delta \cdot\poly(\log(n)+\log(M))$. Similarly, the $\ll$ and $\gg$ notations suppress factors poly-logarithmic in the input length.

We will use the following version of the prime number theorem, where $\pi(n)$ denotes the number of primes in the interval $[n]$ (see, e.g.,~\cite{E49}).

\begin{theorem}[Prime Number Theorem]\label{thm:primeNumber}
For every $\eps>0$, there exists $n_0$ such that for all $n>n_0$, it holds that
\[
(1-\eps)n/\ln(n) \leq \pi(n) \leq (1+\eps)n/\ln(n) \;.
\]
\end{theorem}

\subsection{Function Inversion}
\begin{definition}\label{def:fi}
The Function Inversion problem is a problem to be solved in two phases by a pair of randomized algorithms $\A=(\A_0, \A_1)$. The algorithms receive oracle access to a function $f\colon[N]\to[N']$ (where $N' = \widetilde{O}(N)$), and both $\A_0$ and $\A_1$ can evaluate $f$ at any point $x\in[N]$ in time $\widetilde{O}(1)$.
\begin{description}
\item{\bf Preprocessing phase.} In the first phase, the preprocessing algorithm $\A_0$ preprocesses~$f$ into advice~$\P$ consisting of $S$ bits. 
\item{\bf Query phase.} In the next phase, the online algorithm $\A_1$ receives a query $y \in [N']$ and the advice string~$\P$. If $y \not\in\Im(f)$, then $\A_1$ outputs $\bot$, otherwise $\A_1$ outputs an $x\in f^{-1}(y)$. The running time of the algorithm $\A_1$ is~$T$.
\end{description}
We say that such an algorithm $\A$ for Function Inversion is an $(S, T)$-algorithm
if for every function~$f$, with probability at least $1-1/N$ over the randomness of the algorithms, the online algorithm~$\A_1$ correctly answers \emph{all} queries.

We say that such an algorithm $\A$ for Function Inversion is a \emph{weak} $(S, T)$-algorithm
if for every function $f$, for every query $y \in [N']$ such that $f^{-1}(y)$ is non-empty, the online algorithm $\A_1$ correctly answers $y$
with probability at least $1/2$ over the randomness of the algorithms.
\end{definition}

The classical rigorous algorithm for the Function Inversion problem due to Fiat and Naor~\cite{FN91} solves the problem in space $S$ and time $T$ as long as $S^3 T = \widetilde{\Omega}(N^3)$.
\begin{theorem}[{\cite{FN91}}]\label{thm:FI}
For every $0\leq\delta\leq 1$, there is an $(S,T)$-algorithm for Function Inversion with space $S=\widetilde{O}(N^{1-\delta/3})$ and $T=\widetilde{O}(N^{\delta})$.
\end{theorem}

We will also use another version of this algorithm for Function Inversion that performs better for $T\gg S$.
\begin{theorem}[{\cite{GGPS23}}]\label{thm:FIrecent}
For every $0\leq\delta\leq 1$, there is an $(S,T)$-algorithm for Function Inversion with space $S=\widetilde{O}(N^{1.5-\delta})$ and $T=\widetilde{O}(N^{\delta})$. The algorithm uses $\widetilde{O}(T) = \widetilde{O}(N^{\delta})$ bits of shared randomness.
\end{theorem}

We stress that the online algorithms of \cref{thm:FI,thm:FIrecent} run in time $T$ in the RAM model, but the algorithm of \cref{thm:FIrecent} assumes shared randomness. \cite[Section~6]{GGPS23} proves that \cref{thm:FIrecent} can be implemented without shared randomness at the expense of having computationally unbounded preprocessing and a \emph{non-uniform} online algorithm (as opposed to a RAM online algorithm).  In \cref{thm:sub-functions}, we show that in \emph{our} application, the need for shared randomness can be completely eliminated with no impact on the parameters or the running time of the preprocessing and online algorithms by simply including the shared random string as part of the preprocessed advice. 

We will use the following lemma, which applies standard techniques to amplify the success probability of a weak function inversion algorithm.

\begin{lemma}\label{lem:fworstCaseToAvCase}
Let $f: [N] \to [N']$ be a function such that $N' = \widetilde{O}(N)$. Let $\A=(\A_0, \A_1)$ be a weak $(S, T)$-algorithm for inverting $f$. Then there exists an
$(S', T')$-algorithm for inverting $f$ with $S'=\widetilde{O}(S)$ and $T'=\widetilde{O}(T)$.
\end{lemma}

\begin{proof}
The algorithm $\A' = (\A'_0,\A'_1)$ will create $\ell= \lceil \log(NN') \rceil$ independent copies $\A^0,\ldots,\A^{\ell-1}$ of the assumed weak algorithm, where each copy uses $S$ bits of preprocessing and answers its queries in time~$T$.
\begin{description}
\item {\bf Preprocessing phase:} The advice string produced by $\A'_0$ consists of the $\ell$ advice strings $\P_0,\ldots,\P_{\ell-1}$ produced by $\A^0,\ldots,\A^{\ell-1}$. Clearly, the space complexity of the algorithm is $S'= \ell \cdot S = \widetilde{O}(S)$.
\item {\bf Query phase:} Given a query $y \in [N']$, $\A'_1$ collects the answers $(x^{0},\ldots,x^{\ell-1})$ of the $\ell$ online algorithms $\A^1,\ldots,\A^\ell$. (Some of these answers might be $\bot$, in which case it ignores them.) If for some $i \in [\ell]$, $f(x^{i}) = y$, then $\A'_1$ outputs $x^{i}$, otherwise it outputs $\bot$. The query time of $\A'_1$ is $T'=\widetilde{O}(\ell\cdot T)=\widetilde{O}(T)$.
\item {\bf Analysis:}
It remains to prove that $\A'_1$ answers \emph{all queries} $y \in [N']$ correctly with probability at least $1-1/N$. Since $\A'_1$ verifies that $f(x^{i}) = y$, it never outputs a false positive solution. In~particular, if $f^{-1}(y)$ is empty, then the algorithm always outputs the correct answer $\bot$. Now assume that  $f^{-1}(y)$ is non-empty. Then, the probability that all $\ell$ instances $\A^1,\ldots,\A^\ell$ give wrong answers on~$y$ is at most $1/2^\ell$. Taking a union bound over all $y\in [N']$, the probability that at least one query is not answered correctly is at most 
\[
N' \cdot 2^{-\ell} \leq N' / (NN')= 1/N \;.\qedhere
\]
\end{description}
\end{proof}

\subsection{Data Structure Problems}
\begin{definition}\label{def:ksum}
For a constant integer $k\geq 3$, the \kSUMInd problem is a problem to be solved in two phases by a pair of randomized algorithms $\A=(\A_0, \A_1)$. 
\begin{description}
\item{\bf Preprocessing phase.} In the first phase, the preprocessing algorithm $\A_0$ receives a list of $n$ integers $A=(a_0, \ldots, a_{n-1}) \in \Z_{\geq0}^n$, and preprocesses them into advice~$\P$ consisting of $S$ bits. %

\item{\bf Query phase.} In the next phase, the online algorithm $\A_1$ receives a query $b \in \Z_{\geq0}$ and the advice string~$\P$. If $b \not\in (k-1)A$, then $\A_1$ outputs $\bot$, otherwise $\A_1$ outputs a tuple $(i_1,\ldots,i_{k-1}) \in [n]^{k-1}$ such that  $a_{i_1}+\ldots+a_{i_{k-1}}=b$. The running time of the algorithm $\A_1$ is $T$.
\end{description}
We say that such an algorithm $\A$ for \kSUMInd is an $(S, T)$-algorithm
if for every input $A$, with probability at least $1-1/n$ over the randomness of the algorithms, the online algorithm $\A_1$ correctly answers all queries.
\end{definition}

A few remarks about the definition of \kSUMInd are in order.
\begin{remark}\label{rem:ksumdef}
\begin{enumerate}
\item Since in this work we are concerned with \emph{upper bounds} on the complexity of \kSUMInd, we intentionally choose the weaker computational model for \cref{def:ksum} (which only makes our results stronger). Another standard computational model for this problem is the cell-probe model, where (i) the preprocessing and online algorithms are computationally unbounded, and (ii) $T$ only bounds the number of bits of the advice $\P$ read by the online algorithm $\A_1$. 
\item All randomized algorithms can be implemented as deterministic algorithms in the cell-probe model by fixing the ``best'' randomness in the computationally unbounded preprocessing phase. We remark that the algorithms we present in this work can be derandomized much more efficiently by verifying that all elements of $(k-1)A$ appear in some chain of the Function Inversion subroutine.
\item While the running time of the preprocessing algorithm is not bounded in \cref{def:ksum}, we~note that the preprocessing algorithms presented in this work are efficient and run in time $\widetilde{O}(n^{k-1})$.
\item The success probability of the \kSUMInd algorithm is defined to be at least $1-1/n$, but it can be easily amplified by standard techniques (as in \cref{lem:fworstCaseToAvCase}). 
\item One can assume without loss of generality that the input list of integers $(a_1,\ldots,a_n)\in[M]^n$ satisfies $M=\widetilde{O}(n^{k-1})$. This can be achieved by reducing the inputs $(a_1,\ldots,a_n)$ modulo several primes $p_i=\widetilde{O}(n^{k-1})$ (see, e.g., \cite[Theorem~7]{GGHPV20}).
\end{enumerate}
\end{remark}

Next, we define other data structure problems we are considering in this paper. In order to define the \JI problem, we first need to define the histogram of a string $S\in\Sigma^n$ over an alphabet $\Sigma$. The histogram of~$S$ is a vector $h\in\Z_{\geq0}^{|\Sigma|}$, where the $i$th coordinate of $h$ is the number of occurrences of the $i$th character of the alphabet~$\Sigma$ in the string~$S$.
\begin{itemize}
\item For $k\geq 3$, the \kXORInd problem is a variant of \kSUMInd where the addition is replaced by XOR. The \kXORInd problem takes as input a list of~$n$ vectors $A=(a_0, \ldots, a_{n-1}) \in (\{0,1\}^\ell)^n$. For a query $b \in \{0,1\}^\ell$, the task is to output a tuple $(i_1,\ldots,i_{k-1}) \in [n]^{k-1}$ such that 
$a_{i_1}\oplus\ldots\oplus a_{i_{k-1}}=b$ if such a tuple exists, and to output $\bot$ otherwise.
 
\item For $n\leq m$, the $\ThreeSUMInd(n,m)$ problem is a variant of \ThreeSUMInd where the two input lists may have different lengths. The $\ThreeSUMInd(n,m)$ problem takes as input a list of $n$ non-negative integers $A=(a_0, \ldots, a_{n-1}) \in \Z_{\geq0}^{n}$ and a list of $m$ non-negative integers $B=(b_0, \ldots, b_{m-1}) \in \Z_{\geq0}^{m}$. For a query $b \in \Z_{\geq0}$, the task is to output a pair $(i,j) \in [n] \times [m]$ such that $a_i + b_j = b$ if such a pair exists, and to output $\bot$ otherwise. For $n=m$, $\ThreeSUMInd(n,m)$ is simply the \ThreeSUMInd problem.
\item The \JI problem takes as input a string~$S$ of length~$n$ over an alphabet $\Sigma$. For a query histogram $h\in\Z_{\geq0}^{|\Sigma|}$, the task is to output 1 if there exists a substring of~$S$ with histogram~$h$, and to output~0 otherwise.
\end{itemize}

We will use the following efficient reduction from \JI to \ThreeSUMInd shown by~\cite{BGLPRS24} (the reduction was also implicitly used in \cite{CL15}).
\begin{theorem}[{\cite[Corollary~2]{BGLPRS24}}]\label{thm:JI}
Assume there is an $(S,T)$-algorithm for \ThreeSUMInd. Then there is an $(S', T')$-algorithm for \JI over alphabets of constant size $|\Sigma|=O(1)$, where $S'=\widetilde{O}(S)$ and $T'=\widetilde{O}(T)$.
\end{theorem}

\section{Function Inversion with Sub-Functions}
In \cref{thm:sub-functions}, we describe the structure of $f$ that leads to an improved inversion algorithm for some parameters. In \cref{sec:3sum}, we identify this structure for \ThreeSUMInd and apply \cref{thm:sub-functions}. We present our main result in this modular way because we believe that \cref{thm:sub-functions} may find additional applications for other data structure problems.

Specifically, we present an efficient function inversion algorithm for a class of structured functions~$f\colon[N] \to [N']$. We assume the existence of two efficient mappings, $\MAP_1\colon[N'] \to [D]$ and $\MAP_2\colon[N'] \to [L']$, along with $D$ efficiently computable sub-functions $f_d\colon[L] \to [L']$, such that the following holds with high probability: to invert $f$ at a point $y$, it suffices to invert the function $f_{\MAP_1(y)}$ at the point $\MAP_2(y)$. In this case, the problem of inverting~$f$ reduces to inverting \emph{one} of the $D$ functions $f_d$, each of which has a smaller domain. While this reduction alone does not improve upon the classical Fiat–Naor algorithm, combining it with the space-efficient inversion algorithm from~\cref{thm:FIrecent} allows us to improve the known bounds in this special setting.
\begin{theorem}\label{thm:sub-functions}
Let $f:[N] \to [N']$ be a function such that $N' = \widetilde{O}(N)$ and let $\hat{D} = \widetilde{O}(N),\hat{L} = \widetilde{O}(N),\hat{S} = \widetilde{O}(N)$ be integer parameters. Assume that there is a randomized algorithm $\B$ that takes $f$ as input and outputs $(D,L,L',\AUX)$, where $D,L,L' \in \Z_{\geq0}$ satisfy
$D = \widetilde{O}(\hat{D}),L = \widetilde{O}(\hat{L}),L' = \widetilde{O}(\hat{L})$,
and $\AUX$ is an auxiliary string of length $\widetilde{O}(\hat{S})$ bits, such that the following hold:
\begin{enumerate}
  \item For every $d \in [D]$, there is a deterministic (sub) function $f_d: [L] \to [L']$.
  \item There are deterministic query mapping functions $\MAP_1:[N'] \to [D]$, $\MAP_2:[N'] \to [L']$, 
and an output translation function $\TR:[N'] \times [L] \to [N]$.
  \item \label{constraint:time} All functions $f_d,\MAP_1,\MAP_2,\TR$, given $D,L, L'$, and access to $\AUX$, can be evaluated at any point in time $\widetilde{O}(1)$.  
  \item \label{constraint:1} 
  For every $y \in [N']$ such that $f^{-1}(y)$ is non-empty, 
  with probability at least $5/6$ (over the randomness of $\B$),
  there is $x' \in [L]$ such that 
  $f_{\MAP_1(y)}(x') = \MAP_2(y)$.
  \item 
  \label{constraint:2}
  For every $y \in [N']$ such that $f^{-1}(y)$ is non-empty, 
  with probability at least $5/6$ (over the randomness of $\B$),
  every $x' \in [L]$ such that $f_{\MAP_1(y)}(x') = \MAP_2(y)$ satisfies  
  $\TR(y, x' ) \in f^{-1}(y)$.
\end{enumerate}
Then, for every $0\leq \delta \leq 1$, there is an $(S, T)$-algorithm for inverting $f$ with space $S=\widetilde{O}( \hat{L}^{1.5-\delta} \cdot \hat{D} + \hat{S} + \hat{L}^{\delta})$ and query time $T=\widetilde{O}(\hat{L}^{\delta})$.

In particular, if $\hat{L} = \widetilde{O}(N/\hat{D})$, then 
for every $0\leq \delta \leq 1$, there is an $(S, T)$-algorithm for inverting $f$ with 
$S=\widetilde{O}(N^{1.5-\delta} \cdot \hat{D}^{\delta-0.5} + \hat{S} + N^{\delta}/\hat{D}^{\delta})$ and $T=\widetilde{O}(N^{\delta}/\hat{D}^{\delta})$. 
\end{theorem}
We stress that we assume that the description of $f$ and   
the output $(D,L,L',\AUX)$ of the algorithm $\B$ uniquely define all functions $f_d,\MAP_1,\MAP_2$ and $\TR$.
Depending on the randomness of $\B$, these functions may have different domains and ranges, but they always satisfy $D = \widetilde{O}(\hat{D}),L = \widetilde{O}(\hat{L}),L' = \widetilde{O}(\hat{L})$. We further clarify that all $\widetilde{O}(\cdot)$ and $\widetilde{\Omega}(\cdot)$ notations used in the theorem suppress poly-logarithmic factors in $N$.

\begin{proof}[Proof of Theorem~\ref{thm:sub-functions}]
By~\cref{lem:fworstCaseToAvCase}, it is sufficient to 
describe a weak $(S,T)$-algorithm $\A = (\A_0,\A_1)$ for $f$. The algorithm will use the algorithm of~\cref{thm:FIrecent}, which requires shared randomness. We will first apply this algorithm naively, and then show how to remove the shared randomness with minimal cost.
\begin{description}
\item {\bf Preprocessing phase:}
$\A_0$ runs $\B$ and obtains $(D,L,L',\AUX)$.
For every $d \in [D]$, it runs the preprocessing algorithm of~\cref{thm:FIrecent} for the function $f_d$, 
and obtains the advice $\P_d$.
The advice string is $\P = (D,L,L',\AUX,\P_1,\ldots,\P_D)$.

\item {\bf Query phase:}
On input query $y \in [N']$ and advice string $\P = (D,L,L',\AUX,\P_1,\ldots,\P_D)$,
$\A_1$ first computes $d = \MAP_1(y)$ and $y' = \MAP_2(y)$. 
It then runs the online algorithm of~\cref{thm:FIrecent} for $f_{d}$
with advice string $\P_d$ and query $y'$.
If this algorithm outputs $\bot$, then $\A_1$ outputs $\bot$.
Otherwise, denote the output of that algorithm by $x' \in [L]$.
$\A_1$ then computes  
$x = \TR(y, x')$. If $f(x) = y$, the algorithm outputs $x$, and it outputs $\bot$ otherwise.

Observe that by the definition of $\A$ and by~\cref{thm:FIrecent}, 
for every $0\leq \delta \leq 1$, we have $S=\widetilde{O}( \hat{L}^{1.5-\delta} \cdot \hat{D} + \hat{S})$ 
and $T=\widetilde{O}(\hat{L}^{\delta})$ (since $f_d,\MAP_1,\MAP_2,\TR$ all run in time $\widetilde{O}(1)$).

\item {\bf Analysis:}
Fix $y \in [N']$ for which $f^{-1}(y)$ is non-empty.
$\A_1$ succeeds if the following three events occur simultaneously:
(1) $f_{\MAP_1(y)}^{-1}(\MAP_2(y))$ is non-empty, 
(2) the online algorithm of~\cref{thm:FIrecent} returns 
$x' \in [L]$ such that $f_{\MAP_1(y)}(x') = \MAP_2(y)$, 
and (3) $\TR(y, x' ) \in f^{-1}(y)$.

By~\cref{constraint:1}
the first event occurs with probability at least $5/6$, in which case the second event occurs with probability $1 - 1/L \geq 5/6$
(we may assume that $L \geq 6$, as otherwise, inverting $f_{\MAP_1(y)}$ in constant time is trivial). By~\cref{constraint:2},
the third event occurs with probability $5/6$.

Overall, by a union bound, the success probability of $\A_1$ is at least $1 - 3/6 = 1/2$, hence 
$\A$ is a weak $(S,T)$-algorithm as claimed.

Finally, we deal with the shared randomness of~\cref{thm:FIrecent}. 
For each sub-function $f_d\colon[L]\to[L']$, the length of the shared string is $\widetilde{O}(\hat{L}^{\delta})$. By amplification (\cref{lem:fworstCaseToAvCase}), a shared string of length $\widetilde{O}(\hat{L}^{\delta})$ can guarantee success probability of $1-\widetilde{O}(1/(N\hat{D}))$. Therefore, by a union bound, the same shared random string can be used for all $\widetilde{O}(\hat{D})$ sub-functions $f_d$, and thus, except with probability at most $1/(2N)$, succeed for all of them simultaneously. We now modify the preprocessing algorithm $\A_0$ to send this random string as part of the preprocessing advice. Since this random string is of length $\widetilde{O}(\hat{L}^{\delta})$ bits, we conclude that $S=\widetilde{O}( \hat{L}^{1.5-\delta} \cdot \hat{D} + \hat{S} + \hat{L}^{\delta})$ as claimed.\qedhere
 \end{description}
\end{proof}

\section{Algorithm for \ThreeSUMInd}\label{sec:3sum}
In this section, we present an algorithm for $\ThreeSUMInd(n,m)$ in \cref{thm:threeTSum}. This algorithm will be used to prove our main results for \ThreeSUMInd and its generalization, \kSUMInd.
\begin{theorem}
\label{thm:threeTSum}
For every $0\leq \delta\leq 1$, there is an $(S, T)$-algorithm for $\ThreeSUMInd(n,m)$ with space 
$S=\widetilde{O}(n^{1.5-\delta} \cdot m)$ and query time 
$T=\widetilde{O}(n^{\delta})$.
\end{theorem}

\begin{proof}
    Let $A=(a_0,\ldots,a_{n-1}),B = (b_0,\ldots,b_{m-1})$ be the input of the algorithm, and let $M=\max(\max_i a_i, \,\max_j b_j)$ be the maximum element in the two arrays. We will assume that $M\geq 8$ as otherwise the problem can be trivially solved in constant time and space $S=T=O(1)$ by storing an answer for each query $y \in [2M]$.

    Let $k_1=50n\ln(2M)\ln\ln(2M)=\widetilde{O}(n)$, and $k_2=50m\ln(2M)\ln\ln(2M)=\widetilde{O}(m)$. 
    Let $I_1=\{k_1,\,\ldots, \,2k_1\}$, and $I_2=\{k_2,\,\ldots, \,2k_2\}$. First, we show that the number of primes in $I_1$, denoted by $\pi(I_1)$, is at least $6n\log_{n}(2M)$. A similar argument shows that $\pi(I_2)\geq 6m\log_{m}(2M)$. By the prime number theorem (\cref{thm:primeNumber}) applied with $\eps =1/8$, for all large enough~$n$, we have that

    \begin{align*}
    \pi(I_1) 
    \geq \frac{(1-\eps)2k_1}{\ln(2k_1)}-\frac{(1+\eps)k_1}
    {\ln(k_1)}
    =\frac{k_1}{\ln(k_1)}\cdot\left(1-3\eps-\frac{2(1-\eps)\ln(2)}{\ln(2k_1)}\right)
    \geq\frac{k_1}{2\ln(k_1)}
    \geq6n\log_{n}(2M)\;.
    \end{align*}
    Here the penultimate inequality uses $k_1\geq 2^{13}$ (which holds for all large enough $n$), and the last inequality uses $\ln\ln(2M)\geq 1$ and $n\geq 50$.
    
    We now prove the result using~\cref{thm:sub-functions}. 
    First, let $N = n \cdot m$, $N' = 2M$, and let us define $f:[N] \to [N']$ as
    $f(i,j) = a_i + b_j$.
    Let $\hat{D} = m,\; \hat{L} = n,\; \hat{S} = m$. 
      
    Since $\hat{L} = n = \widetilde{O}((n \cdot m)/\hat{D})$, assuming~\cref{thm:sub-functions} applies with these parameters, we have that for every $0\leq \delta \leq 1$, 
    there is an $(S, T)$-algorithm for $\ThreeSUMInd(n,m)$ with space
    \begin{align*}
    S&=\widetilde{O}(N^{1.5-\delta} \cdot \hat{D}^{\delta-0.5} + \hat{S} + N^\delta/\hat{D}^\delta)\\
    &=\widetilde{O}((n \cdot m)^{1.5-\delta} \cdot (m)^{\delta-0.5} + m +n^\delta)\\
    &=\widetilde{O}(n ^{1.5-\delta} \cdot m + m+n^\delta)\\
    &=\widetilde{O}(n ^{1.5-\delta} \cdot m)
    \end{align*}
    (recalling that $m \geq n$) and query time
    \[T=\widetilde{O}(N^\delta/\hat{D}^\delta)
    =\widetilde{O}(n^\delta)\;,\]
    proving the theorem.
   
    We now define the algorithms and functions required to apply~\cref{thm:sub-functions}. 
    \begin{description}
\item {\bf Algorithm $\B$:} 
    $\B$ works as follows: 
    It draws $p$ from $I_1$ and $q$ from $I_2$ as independent and uniformly random primes,
    and sets $L = n,L' = p$ and $D = q$. 
   $\AUX$ contains:
    \begin{enumerate}
      \item The $n$ input numbers~$A=(a_0,\ldots,a_{n-1})$ sorted in the non-decreasing order (and in the increasing order of~$i$ in case of ties);
      \item The $m$ input numbers~$B=(b_0,\ldots,b_{m-1})$ sorted in the non-decreasing order (and in the increasing order of~$j$ in case of ties); 
      \item The $m$ input numbers~$B=(b_0,\ldots,b_{m-1})$ sorted in the non-decreasing order of their remainders modulo~$q$ (and in the increasing order of~$j$ in case of ties);          
      \item A uniformly chosen integer $z\in[p]$. 
    \end{enumerate}   

    \item {\bf Functions $f_d\colon[n]\to[p]$ for $d\in[q]$:} 
    For every $d\in[q]$, define the function $f_d\colon[n]\to[p]$ as follows. For every $i\in[n]$,
    \[
        f_d(i) = \begin{cases}
			a_{i}+b_j \bmod p, & \text{for minimum~$j$ such that $a_{i}+b_j \equiv d \bmod q$ if such $j\in[m]$ exists;}\\
            z, & \text{if $a_{i}+b_j\not\equiv d \bmod q$ for all $j\in[m]$.}
		 \end{cases}
    \]
    Observe that given the list of input numbers $B$ sorted in the non-decreasing order by their remainders modulo~$q$ (in $\AUX$), the function $f_d(i)$ is indeed computable in time $\widetilde{O}(1)$ by binary search 
    for $d - a_{i} \bmod q$, as required by~\cref{constraint:time} of~\cref{thm:sub-functions}.
    
    \item {\bf Functions $\MAP_1:[2M] \to [q]$, $\MAP_2:[2M] \to [p]$:}
    define $\MAP_1(y) = y \bmod q$ and  $\MAP_2(y) = y \bmod p$. Clearly, these functions are computable in time $\widetilde{O}(1)$, as required by~\cref{constraint:time} of~\cref{thm:sub-functions}.
    
    \item {\bf Function $\TR:[2M] \times [n] \to [n] \times [m]$:}
    Given $y \in [2M]$ and $i \in [n]$, $\TR(y, i)$
    uses binary search to find a $b_j$ such that $a_{i}+b_j=y$, and returns $(i ,j)$
    (if no such $j$ is found, it returns an arbitrary pair, e.g., $(1,1)$).
    Clearly, $\TR$ is computable in time $\widetilde{O}(1)$ given the sorted list of numbers, as required by~\cref{constraint:time} of~\cref{thm:sub-functions}.
    
    \item {\bf Analysis:} It remains to prove that the defined functions satisfy the constraints of~\cref{constraint:1} and~\cref{constraint:2} of~\cref{thm:sub-functions}.

    We first prove the constraint of~\cref{constraint:1}.
    Fix $y \in [2M]$ such that $y \in A + B$, namely, there exist $(i,j) \in [n] \times [m]$ such that 
    $a_i + b_j = y$. Fix such $(i,j)$.

    We claim that 
    \begin{align} \label{eq:0}
     \Pr[\exists j' \in [m] \colon  b_{j'} \neq b_{j} \wedge b_j \equiv b_{j'} \bmod q] \leq 1/6.  
    \end{align}
    This implies that $\Pr[f_{y \bmod q}(i) = (a_i + b_j \bmod p) =  (y \bmod p)] \geq 5/6$,
    hence the constraint of~\cref{constraint:1} holds.

    To prove~\cref{eq:0}, observe that $b_{j'}-b_j \in \{-2M,\ldots,2M\} \backslash \{0\}$, hence $b_{j'}-b_j$ has at most $\log_{m}(2M)$ distinct prime factors from $I_2$. For a fixed $b_{j'}$, the probability (over the choice of~$q$) that $b_{j'}-b_j \equiv 0 \bmod q$ is then at most $\log_{m}(2M)/\pi(I_2)\leq \frac{1}{6{m}}$. Taking a union bound over all $j' \in [m]$, such that $b_{j'} \neq b_j$ we have that $b_j \not\equiv b_{j'} \bmod q$ for all $j'$ such that $b_{j'} \neq b_j$ with probability at least $5/6$.

    Next, we prove the constraint of~\cref{constraint:2}. Again, fix $y \in A + B$.
    We claim that 
    \begin{align}
    \Pr[y \not\equiv z \bmod p] &\, = 1- 1/p \geq\, 1-1/{n} \;, \label{eq:one} \\
    \Pr[| \{ y'\in A+B, y \equiv y' \bmod pq \} | = 1] &\,\geq\, 35/36. \label{eq:two}
    \end{align}

    Given that $y \not\equiv z \bmod p$, if $f_{y \bmod q}(i') = y \bmod p$ for some $i' \in [n]$,
    then there exists $j' \in [m]$ such that $a_{i' } + b_{j'} \equiv y \bmod pq$.
    Furthermore, from 
    $| \{ y'\in A+B, y \equiv y' \bmod pq \} | = 1$, we have that 
    $a_{i'} + b_{j'} = y$,
    implying that $\TR(y, i' ) = (i',j')$.
    
    Therefore, these two inequalities imply that the constraint of~\cref{constraint:2} holds with probability at least $35/36 - 1/{n} \geq 5/6$ (for sufficiently large $n$), as required. 

    It remains to prove the two inequalities.
    The proof of~\cref{eq:one} is trivial by the uniform choice of $z \in [p]$. 
    
    The proof of~\cref{eq:two} is similar to the proof of~\cref{eq:0}: let $y' \in A + B$ such that $y' \neq y$.
    Since $y'-y \in \{-2M,\ldots,2M\} \backslash \{0\}$, we have that $y'-y$ has at most $\log_{n}(2M)$ distinct prime factors from $I_1$ and at most $\log_{m}(2M)$ distinct prime factors from $I_2$. For a fixed $y'$, the probability (over the choice of $p$ and $q$) that $y'-y \equiv 0\bmod pq$ is then at most $(\log_{n}(2M)/\pi(I_1)) \cdot (\log_{m}(2M)/\pi(I_2)) \leq \frac{1}{36 (n \cdot m)}$. Taking a union bound over all $y'\in A+B$, such that $y' \neq y$ we have that $y \not\equiv y' \bmod pq$ for all $y' \neq y$ with probability at least $35/36$. \qedhere
    \end{description}
\end{proof}

We are now ready to conclude \cref{thm:threeSum} from \cref{thm:threeTSum}.
\ThreeSumThm*
\begin{proof}
By setting $m = n$ in~\cref{thm:threeTSum}, we obtain an $(S, T)$-algorithm for \ThreeSUMInd with space 
$S=\widetilde{O}(n^{2.5-\delta})$ and query time $T=\widetilde{O}(n^{\delta})$.
\end{proof}

\section{Generalizations and Applications}
We begin this section with a generalization of the algorithm for \ThreeSUMInd to \kSUMInd.
\KSumThm*
\begin{proof}
Let $A=\{a_0,\ldots,a_{n-1}\}$ be an instance of \kSUMInd.
We define a $\ThreeSUMInd(n,m)$ instance with input $(A,B)$ by setting $B = (k-2)A$. Clearly, each query $b$ to \kSUMInd on input~$A$, is equivalent to the same query to $\ThreeSUMInd(n,m)$ on input~$(A, B)$. The constructed instance of $\ThreeSUMInd(n,m)$ has $m = O(n^{k-2})$.
Applying~\cref{thm:threeTSum} to this instance, we deduce that  
for every $0\leq \delta\leq 1$, 
there is an $(S, T)$-algorithm for \kSUMInd with space 
\[S=\widetilde{O}(n^{1.5-\delta} \cdot m) = 
\widetilde{O}(n^{k-0.5-\delta})
\]
and query time 
\[T=\widetilde{O}(n^{\delta})\;,\]
which completes the proof.
\qedhere
\end{proof}
In the next theorem, we show that one can modify the proofs of \cref{thm:threeTSum,thm:kSum} to obtain an analogous result for the \kXORInd problem.
\kXORThm*
\begin{proof}[Proof sketch]
The main difference between the proofs of this theorem and \cref{thm:kSum} is in the way we define the functions $\MAP_1, \MAP_2, \TR$, and $f_d$ when applying \cref{thm:sub-functions}. 
Let $A=(a_0,\ldots,a_{n-1})$ be the input of the algorithm, where each $a_i\in\F_2^\ell$ is viewed as a vector from $\F_2^\ell$. Let $B=(b_0,\ldots,b_{m-1})=(k-2)A$, where the addition, again, is over $\F_2^\ell$, and $m=O(n^{k-2})$.  We will assume that $\ell\geq1.5\log n$ as otherwise the problem can be trivially solved in constant time and space $S=\widetilde{O}(n^{1.5})=\widetilde{O}(n^{k-0.5-\delta})$ by storing an answer for each query $y \in \F_2^\ell$.

The algorithm $\B$  samples independent uniformly random full-rank matrices $P\in\F_2^{p\times \ell}$ and $Q\in\F_2^{q \times \ell}$ for $p=\log{n}+O(1)$ and $q=\log{m}+O(1)$. Then $\MAP_1\colon \F_2^\ell\to\F_2^q$ and $\MAP_2\colon \F_2^\ell\to\F_2^p$ are defined as
\begin{align*}
\MAP_1(y) &= Q y\;,\\
\MAP_2(y) &= P y\;.
\end{align*}
Let $z\in\F_2^p$ be a uniformly random vector. For every $d\in\F_2^q$, we define the function $f_d\colon[n]\to\F_2^p$ as follows. For every $i\in[n]$,
    \[
        f_d(i) = \begin{cases}
			P(a_{i} \oplus b_j), & \text{for minimum~$j$ such that $Q(a_{i} \oplus b_j) = d $ if such $j\in[m]$ exists;}\\
            z, & \text{if $Q(a_{i} \oplus b_j)\neq d$ for all $j\in[m]$.}
		 \end{cases}
    \]
Note that each $f_d$ can be efficiently computed given a sorted list of $Qb_j$ for all $j\in[m]$.
Finally,  Given $y \in \F_2^\ell$ and $i \in [n]$, $\TR(y, i)$
    uses binary search to find a $b_j$ such that $a_{i} \oplus b_j=y$, and returns $(i ,j)$
    (if no such $j$ is found, it returns an arbitrary pair, e.g., $(1,1)$).
    $\TR$ is computable in time $\widetilde{O}(1)$ given a sorted list of $b_j$.
We now apply \cref{thm:sub-functions} with the only difference that the ranges and domains of the functions are not subsets of integers but vector spaces over $\F_2$. It is now not hard to see that analogues of \cref{eq:0,eq:one,eq:two} hold for the functions defined above.
\end{proof}

Finally, we present an immediate application of the \ThreeSUMInd algorithm from \cref{thm:threeSum} to the \JI problem.

\JICor*
\begin{proof}
    This result immediately follows from the reduction from \JI over constant-size alphabets to \ThreeSUMInd of \cite{BGLPRS24} (\cref{thm:JI}) and the \ThreeSUMInd algorithm from \cref{thm:threeSum}.
\end{proof}

\section*{Acknowledgments}
We thank the anonymous referees for their very helpful comments.
\newpage
\bibliographystyle{alpha}
\bibliography{refs.bib}

\end{document}